\documentclass{new_tlp}

\usepackage{graphicx}
\usepackage{hyperref}
\usepackage{cmll}
\usepackage{proof}
\usepackage{tikz}
\usepackage{color}
\usepackage{stmaryrd}
\usepackage{latexsym}
\usepackage{amsmath}
\usepackage{amssymb}
\usepackage{comment}
\usepackage{color}
\usepackage{url}
\usepackage{mycommands}
\usepackage{lineno}
\usepackage{tweaklist}

\newcommand\bcmdtab{\noindent\bgroup\tabcolsep=0pt%
  \begin{tabular}{@{}p{10pc}@{}p{20pc}@{}}}
\newcommand\ecmdtab{\end{tabular}\egroup}

  \title[A Proof Theoretic Study of Soft CCP]
          {A Proof Theoretic Study of Soft Concurrent Constraint Programming}

  \author[Elaine Pimentel,  Carlos Olarte and Vivek Nigam]
         {         ELAINE PIMENTEL\\
         Universidade Federal do Rio Grande do Norte, Natal, Brazil\\
         \email{elaine.pimentel@gmail.com}
          \and 
         CARLOS OLARTE\\
         Pontificia Universidad Javeriana Cali, Colombia\\
                  Universidade Federal do Rio Grande do Norte, Natal, Brazil\\
         \email{carlos.olarte@gmail.com}
         \and
         VIVEK NIGAM\\
         Universidade Federal da Para\'iba, Jo\~ao Pessoa, Brazil\\
         \email{vivek.nigam@gmail.com}\\
         }

\begin{document}
\label{firstpage}

\maketitle

\begin{abstract}
Concurrent Constraint Programming (CCP) is a simple and powerful model for concurrency 
where agents interact by telling and asking constraints. Since their inception, CCP-languages have been designed for having a strong connection to logic. In fact, 
the underlying constraint system can be built from a suitable fragment of intuitionistic (linear) logic --ILL-- and processes can be interpreted as formulas in ILL. Constraints as ILL formulas fail to represent accurately situations where ``preferences'' (called soft constraints) such as probabilities, uncertainty or fuzziness are present. In order to circumvent this problem, c-semirings have been proposed as algebraic structures for defining constraint systems where agents are allowed to tell and ask soft constraints. Nevertheless, in this case, the tight connection to logic and proof theory is lost. In this work, we give a proof theoretical meaning to soft constraints: they can be defined as formulas in a suitable fragment of ILL with subexponentials (SELL) where subexponentials, ordered in a c-semiring structure, are interpreted as preferences. We hence achieve two goals: (1) obtain a CCP language where agents can tell and ask soft constraints and (2) prove that the language in (1) has a strong 
connection with logic. Hence we keep a declarative reading of processes as formulas while providing a logical framework for soft-CCP based systems. An interesting side effect of (1) is that one is also able 
to handle probabilities (and other modalities) in SELL, by restricting the use of the promotion rule for non-idempotent c-semirings.
This finer way of controlling subexponentials allows for considering more interesting spaces and restrictions, and it opens the possibility of specifying more challenging computational systems.
\end{abstract}

  \begin{keywords}
   Concurrent Constraint Programming, Linear Logic, Soft Constraints
  \end{keywords}

\section{Introduction}

Providing logical and proof theoretic semantics to (fragments of) programming languages not only 
gives a declarative meaning to these languages, but also normally leads to 
 the development of new features allowing more expressive programming constructions to be used. For example,
we investigated recently~\cite{DBLP:conf/concur/NigamOP13} a proof theoretic
specification of the concurrent constraint  programming (\ccp) \cite{saraswat91popl} language introduced in \cite{knight12concur} 
that mentions epistemic (\eccp) and spatial (\sccp) modalities. We used as underlying logical framework  
linear logic with subexponentials (\sell)~\cite{danos93kgc,nigam09ppdp}, showing that our encodings 
faithfully specify \eccp\ and \sccp. More interestingly, this study allowed us to develop  
extensions of \eccp\ and \sccp\ with features not available in \cite{knight12concur}, such as systems with an unbounded number 
of agents for \eccp\ or spaces for \sccp\ and constructs that allow the communication of location names~\cite{olarte13lsfa}.
In this paper we turn our attention to 
computing with \emph{soft constraints}.

Soft concurrent constraint programming (Soft-\ccp) \cite{DBLP:journals/tocl/BistarelliMR06} is an extension of \ccp\
where agents are allowed to tell and ask soft constraints, \ie, constraints with certain level of preference. 
This allows the  modeling of systems with
levels of uncertainty, or those mentioning probabilistic or fuzzy entities. 
However, moving from hard (crisp) constraints   to soft constraints was not followed
by a corresponding logical/proof theoretic characterization of these systems. This is unfortunate
because one of the key motivations of the original \ccp\ was its tight connection to logic and proof theory which 
enabled the proposal of more advanced systems such as its linear version \lcc~\cite{fages01ic}.
The main contribution of this paper is to recover this connection by studying the proof theory of soft constraint systems 
in the form of \sell\ theories.

A key feature of 
\sell\ is that subexponentials  are organized into a pre-order, 
specifying the provability relation among subexponentials. In our previous work~\cite{DBLP:conf/concur/NigamOP13}, 
we used a simple partially ordered set of subexponential names. While this was enough for modeling \eccp\ and \sccp, 
it is not enough to model soft constraints. In this case, we need more sophisticated algebraic structures. 

This paper investigates the proof theory of \sell\ with more involved pre-orders and demonstrates
that it is possible to characterize \emph{soft constraint systems}  by using pre-orders based on general semiring definitions. 
During this investigation, we have also identified variants of soft constraints, namely those based on non-idempotent 
c-semirings, such as probabilistic soft constraints, that do not seem to have a sensible correspondence in \sell.
We thus propose a new proof system, called \sells, with a different promotion rule that allows for such a correspondence.
We prove that \sells\ admits cut-elimination.
We also point out that due to the tight correspondence of soft constraints and proof theory, 
it seems possible to extend \eccp\ and \sccp\ with 
soft constraints. Although we provide some pointers in this paper, its full development is left as future work.

\noindent {\bf Organization.}   Section~\ref{sec:ccp} reviews the main machinery on \ccp\ and soft constraints.
Then in Section~\ref{sec-semiring-sells}, after reviewing \sell, we propose a novel soft constraint system based on 
subexponential signatures proving that it is indeed a sensible \ccp\ constraint system (Theorem~\ref{thm:constraints}). 
Then we also  propose an encoding of Soft-\ccp\ into \sell\ proving its adequacy (Theorem~\ref{thm:adeq-basic}).
Section~\ref{sec:computing} 
gives some examples of the use of the  novel soft constraint system and we point out some limitations of \sell\ to represent non-idempotent soft constraints (e.g.,  probabilistic systems). We
thus propose \sells\ and prove that it admits cut-elimination (Theorem~\ref{thm:sells-cut}). 
Section~\ref{sec:concluding} concludes and presents related works. 
Some missing proofs and auxiliary results are shown in the Appendix.

\section{Concurrent Constraint Programming}
\label{sec:ccp}
Concurrent Constraint Programming (\ccp) \cite{saraswat91popl}  (see a survey in \cite{DBLP:journals/constraints/OlarteRV13}) is  a model for 
concurrency that combines the traditional operational view of process calculi 
with a \emph{declarative} view based on logic.  This allows \ccp\ to benefit from 
the large set of reasoning techniques of both process calculi and logic.   
 
Processes in \ccp\ \emph{interact} with each other by \emph{telling} and \emph{asking} 
constraints (pieces of information) in a common 
store of partial information. The type of constraints
processes may act on  is not fixed but parametric in a constraint system (CS). 
 Intuitively, a
CS  provides a signature from which  constraints
can be built from basic tokens (e.g., predicate symbols), and two basic operations: conjunction ($\sqcup$) and variable hiding ($\exists$). The CS defines also an
\emph{entailment} relation ($\entails$) specifying inter-dependencies
between constraints:  $c\entails d$  means that the
information $d$ can be deduced from the information 
$c$. 
Such systems  can be formalized as a Scott information system as in \cite{saraswat91popl}, 
 or they can be built upon a suitable fragment of logic \eg, as in \cite{fages01ic}. 
 In Section~\ref{sec:semisub}, we will 
 specify such systems  as formulas in  
 intuitionistic linear logic (ILL~\cite{girard87tcs}).

\vspace{-2mm}
\subsection{The language of \ccp\ processes}
In the spirit of process calculi, the language of processes in   \ccp\  is given by a small number of primitive operators or combinators as described below. 
\begin{definition}[Syntax. Indeterminate \ccp\ language \cite{saraswat91popl}]\label{def:syntax-lcc} 
 Processes in \ccp\  are built from constraints in the underlying CS and  the syntax:
\[
P,Q ::= \tellp{c} \mid   
\sum\limits_{i\in I}\whenp{c_i}{P_i} \mid   
P \parallel Q  \mid 
\localp{x}{P}   \mid  
p(\overline{x})
\]
\end{definition}

The process $\tellp{c}$ adds $c$ to the current store $d$ producing the new store $c\sqcup d$.
 Given a non-empty finite set of indexes $I$, the process $\sum\limits_{i\in I}\whenp{c_i}{P_i}$ non-deterministically chooses  $P_j$ for execution if the   store entails $c_j$. The chosen alternative, if any, precludes the others.
This provides a powerful synchronization mechanism based on constraint entailment.   
 When $I$ is a singleton, we shall omit the ``$\sum$'' and we simply write $\whenp{c}{P}$. 
The process $P\parallel Q$ represents the parallel (interleaved) execution of $P$ and $Q$. 
The process 
$\localp {x}{P}$ behaves as $P$ and binds the variable 
$x$ to be local to it. We shall use $\fv(P)$ to denote the set of free variables of $P$. 
 Given a process definition  $p(\overline{y}) \defsymboldelta P$,  
where all free variables of $P$ are in the set of pairwise distinct
variables $\overline{y}$, the process $p(\overline{x})$  evolves into $P[\overline{x}/\overline{y}]$.  A \ccp\ program takes then the form $\cD.P$ where $\cD$ is a set of  process definitions and $P$ is a  process. It is assumed that any process name has a unique definition in $\cD$. \\

\noindent{\bf Structural Operational Semantics (SOS)}
The SOS  of \ccp\ is given by the transition relation $ \gamma \redi \gamma'$  
satisfying the rules in Figure \ref{fig:sos}.
Here we follow the formulation 
 in \cite{fages01ic,DBLP:conf/fsttcs/HaemmerleFS07} where the local variables created by the program appear explicitly in the transition system and parallel composition of agents is identified as a multiset of agents. 
More precisely, a \emph{configuration} $\gamma$ is a triple of the  form  
$(X;  \Gamma ;  c)$, where $c$ is a constraint representing the  store,  $\Gamma$ is a multiset of processes,
and $X$ is a set of hidden 
(local) variables of $c$ and $\Gamma$. The multiset $\Gamma=P_1,P_2,\ldots,P_n$  
represents the process  $P_1 \parallel P_2 \parallel \cdots \parallel P_n$. We shall indistinguishably
use both notations to denote parallel composition. Moreover, processes  are quotiented by a structural congruence relation $\equivP$  satisfying: 
 (STR1) $\localp{x}{P} \equivP \localp{y}{P[y/x]}$ if $y \notin fv(P) $   (alpha conversion);
 (STR2) $P\parallel Q \equivP Q \parallel P$;
 (STR3) $P \parallel (Q \parallel R) \equivP (P \parallel Q)  \parallel R$. We shall write $(X;\Gamma;c) \equiv (X';\Gamma';c')$ whenever $X=X'$, $\Gamma \equivP \Gamma'$ and $c\equiv c'$ (i.e., $c\entails c'$ and $c'\entails c$).

The rules in Figure \ref{fig:sos}  are straightforward realizing the operational intuitions given above: a tell agent $\tellp{c}$ adds $c$ to the current store $d$ (Rule $\rTell$); the process $\sum\limits_{i\in I}\whenp{c_i}{P_i}$ executes $P_j$ if its corresponding guard $c_j$ can be entailed from the store (Rule $\rSum$); a local process $\localp{x}{P}$ adds $x$ to the set of hidden variable $X$ when no clashes of variables occur (Rule $\rLocal$). Observe that Rule $\rEquiv$ can be used, for instance,  to do alpha conversion if the premise 
of $\rLocal$
 cannot be satisfied; the call $p(\vx)$ executes the body of the process definition (Rule $\rCall$).

\begin{definition}[Observable behavior]
Let $\redirex$ be the reflexive and transitive closure of $\redi$. If 
$(X;\Gamma; d) \redirex(X';\Gamma';d')$ and 
$\exists X'. d' \entails c$ we write 
$\Barb{(X;\Gamma;d)}{c}$.
If $X=\emptyset$ and $d=\true$   we simply write   $\Barb{\Gamma}{c}$.
\end{definition}

Intuitively, if $P$ is a process then 
$\Barb{P}{c}$ says that $P$ can reach a store $d$ strong enough to entail $c$, \ie, we can regard $c$ as an output of $P$. Note that in the above definition, the variables in $X'$ are hidden  since the information about  them is not observable.

As processes manipulate the store of constraints, the  CS used dictates much of the behavior of
the system.  For instance, in  \cite{fages01ic} it is shown  that 
by using formulas in a fragment of ILL~\cite{girard87tcs}   as CS, one obtains a more expressive language called Linear Concurrent 
Constraint (\lcc) where ask processes can \emph{consume} information 
from the {store}. 
The same goal is achieved here, but by demonstrating that 
soft constraints in 
\ccp\  can be 
obtained by allowing subexponentials \cite{danos93kgc} in the CS.

\begin{figure}
\resizebox{\textwidth}{!}
{
$
\begin{array}{ccc}
\infer[\rTell]
{(X; \tellp{c},\Gamma;d) \redi (X;\Gamma;c\sqcup d)}
{}
\qquad
\infer[\rSum]{
  (X;\sum\limits_{i\in I}\whenp{c_i}{P_i} ;\Gamma;d) \redi (X;P_j;d)
  }
  {
  d \entails c_j \quad j\in I
  }
 \\\\
\infer[\rLocal]
{(X;\localp{x}{P},\Gamma;d) 
\redi (X\cup\{x\};P,\Gamma;d)
}
{x \notin X \cup fv(d) \cup fv(\Gamma)}
\qquad
\infer[\rCall]
{(X;p(\vy),\Gamma;d) \redi  (X;P\sxy,\Gamma;d) }
{p(\vx) \defsymboldelta P  \in \cD } \\\\
\infer[\rEquiv]
{(X;\Gamma;c) \redi (Y;\Delta;d)}
{(X;\Gamma;c) \equiv (X';\Gamma';c') \redi (Y';\Delta';d') \equiv (Y;\Delta;d)}
\end{array}
$
}
\vspace{-2mm}
\caption{Operational semantics for \ccp\ calculi\label{fig:sos}}
\vspace{-2mm}
\end{figure}



\vspace{-2mm}
\subsection{Soft  Constraint in Concurrent Constraint Programming}
\label{sec:subexp-soft}
\vspace{-2mm}
It is well known that crisp (hard)  constraints fail to represent accurately situations where soft constraints, \ie, preferences, probabilities, uncertainty or fuzziness, are present. In constraint programming  \cite{cp-handbook}, two general frameworks have been proposed to deal with soft constraints: \emph{semiring} based constraints \cite{DBLP:journals/jacm/BistarelliMR97} and \emph{valued} constraints \cite{DBLP:conf/ijcai/SchiexFV95}. Roughly speaking, in both frameworks an algebraic structure defines the operations needed to \emph{combine} soft constraints and 
choosing when a constraint (or solution) is \emph{better} than another. 
In \cite{DBLP:journals/constraints/BistarelliMRSVF99}, it is shown that both frameworks are equally expressive and they are general enough to represent different kind of soft constraints including, e.g.,  fuzzy, probabilistic and weighted constraints. 

In the forthcoming sections, we shall 
build soft constraints from formulas in a suitable fragment of ILL with subexponentials (SELL) where subexponentials are ordered in a semiring structure.  Before that, let us  recall  the framework of semiring based constraints.

\begin{definition}[C-Semiring  \cite{DBLP:journals/jacm/BistarelliMR97}] \label{def:c-semiring}
A c-semiring is a tuple $\langle \cA,+_\cA,\times_\cA, \botA,\topA\rangle$ satisfying:
 (S1) $\cA$ is a set and $\botA,\topA \in \cA$;
 (S2) $+_\cA$ is a binary, commutative, associative and  idempotent operator on $\cA$,
 $\botA$ is its  unit element  and $\topA$ its absorbing element;
 (S3) $\times_\cA$ is a binary, associative and  commutative operator on $\cA$ with 
 unit element $\topA$ and absorbing element
 $\botA$.  Moreover, $\times_{\cA}$  distributes over $+_\cA$.
Let $\leqA$ be defined as $a \leq_{\cA} b$ iff $a+_\cA b = b$. Then,  $\langle \cA, \leqA\rangle$ is a complete lattice  where:
 (S4) $+_\cA$ and $\times_\cA$ are monotone on $\leqA$;
 (S5) $\times_\cA$ is intensive on $\leqA$, \ie, $a \times b \leqA a$.
 (S6) $\botA$ (resp. $\topA$) is the bottom (resp. top) of $\cA$; 
 (S7) $+_\cA$ is the $\lub$ operator. 
If $\times_\cA$ is idempotent, then:
 (S8) $+_\cA$ distributes over $\times_\cA$;
 (S9) $\langle \cA,\leqA \rangle$ is a complete distribute lattice and $\times_\cA$ is its $\glb$.
We shall say that a c-semiring is idempotent whenever its $\times_\cA$ operator is idempotent, and non-idempotent otherwise.
\end{definition}

 Elements in the set $\cA$ (c-semiring values) are used to denote the {\em upper bound of preference degrees}, or simply {\em preference level}, 
where the ``preference'' could be a probability, cost, etc. The $\times_{\cA}$ operator is used to combine values 
while $+_{\cA}$ is used to select which is the ``best'' value   in the sense that $a+_\cA b = b$ iff $a \leqA b$ iff $b$ is ``better'' than $a$.\\

\noindent {\bf Instances of c-semirings}
Before giving some instances of c-semirings, an important clarification is in order. 
In soft constraint logic programming \cite{DBLP:journals/jacm/BistarelliMR97} and soft concurrent constraint programming \cite{DBLP:journals/tocl/BistarelliMR06}, constraints are usually seen as mappings from variable assignments into 
elements in the semiring $\cA$. 
For instance, let $x$ and $y$  
be integer variables and consider the constraint 
$\texttt{leq}(x,y)$ with the usual meaning. Then, using the crisp semiring described below,  the constraint $\texttt{leq}(x,y)$
 maps the tuple $\langle1,2\rangle$ to $\true$ and $\langle 2,1\rangle$ to $\false$. 
  Hence, combining two constraints $c_1$ and $c_2$ means that there are \emph{fewer} possible values  in the variable domains that can
  satisfy both constraints (i.e., the variable-assignment problem is ``harder'' to solve). 
 In this paper  we adhere to the tradition of \ccp-languages and constraint systems~ \cite{saraswat91popl,BoerPP95} where constraints are seen as tokens  of (partial) information. Hence, when the token  $\texttt{leq}(x,y)$ is added to the current store $d$,  we are not interested in \emph{solving} the constraint problem $d\sqcup \texttt{leq}(x,y)$ (i.e., find the values for $x$ and $y$ that satisfy such constraint). Instead, we see the  addition of  $\texttt{leq}(x,y)$ to $d$ as   \emph{increasing monotonically} the information we have about $x$ and $y$ in $d$. For instance, that information can be used to deduce (via the entailment relation) that $\texttt{leq}(x,y+1)$ also holds. Accordingly, in the context of soft constraints, adding a constraint $c$ with a \emph{preference level} $a\in \cA$, denoted as 
 $\softC{c}{a}$, will mean that $c$ is \emph{believed}   with 
 a probability, preference, costs, etc.  $a$. The higher the value of $a$ the more the information we add to the store. 
 
 Let us now   give some well-known instances of c-semirings. 
Let $c_1$ and $c_2$ be constraints. 
The c-semiring $S_c = \langle \{\true,\false\}, \vee,\wedge,\false,\true \rangle$ models  
{\bf crisp}  (hard) constraints. 
Then, $\softC{c_1}{\false}$ 
means that the agent \emph{does not believe} in $c_1$ and hence, regardless the preference level   of $c_2$, the conjunction of $c_1$ and $c_2$ must be also assigned a preference level of $\false$. The {\bf fuzzy}  c-semiring $S_F = \langle [0,1], max, min, 0,1\rangle$ 
allows  for fuzzy constraints 
that  have an  associate   preference level  in the real interval $[0,1]$   where 1 represents the best value.
Then, if $\softC{c}{0.2}$ and $\softC{d}{0.7}$
are in the store, we can say that $d$ is believed with a ``better'' (higher) preference level (wrt $+_\cA$) than $c$. 
From that store we can also deduce that  the conjunction $c\sqcup d$ is believed with preference level $0.2$ (using the $\times_\cA$ operator to combine 0.7 and 0.2). In  a {\bf probabilistic} setting,
 a   constraint $c$ is annotated with its probability of existence where probabilities are supposed to be independent (\ie, no
 conditional probabilities). 
 This can be modeled with the c-semiring $S_P = \langle[0,1], max, \times, 0 ,1 \rangle$. Then, if $\softC{c_1}{0.2}$ and $\softC{c_2}{0.7}$ are in the store,  the probability of deducing $c_1 \sqcup c_2$  is $0.14$.  In {\bf weighted} constraints there is an accumulate cost that can be  computed  with the c-semiring $S_w = \langle\cR^-,max,+,-\infty,0 \rangle$, where $0$ means no cost. 
 Then, from a store containing  $\softC{c_1}{-2}$ 
 and  $\softC{c_2}{-7}$ we can deduce $\softC{c_1 \sqcup c_2}{-9}$. 
We note that the first two c-semirings are idempotent (\ie, $\times_\cA$ idempotent), while the last two are not.

\section{Soft-\ccp\ as Theories in Linear Logic with Subexponential}
\label{sec:semisub}

In this section we build soft constraints from formulas in a suitable fragment of 
 intuitionistic linear logic  (ILL) with subexponentials  \cite{danos93kgc,nigam09ppdp} 
(\sell) where subexponentials are ordered in a c-semiring structure. 
By doing that, we achieve two goals: (1) obtain a \ccp\ language where agents can tell and ask soft constraints and (2) prove that the language in (1) has a strong connection with logic (Section \ref{sec:logical-reading}). We then  keep a declarative reading of processes as formulas and provide a logical framework
for soft-CCP based systems.
This last goal is remarkable.
In fact, the beauty of CCP relies on the fact that it is simple, yet powerful, and with a strong connection to logic, hence correct.

\vspace{-2mm}
\subsection{Linear Logic with Subexponentials}
\ \sell\ shares with intuitionistic linear logic  all 
its connectives except the exponentials:
instead of having a single pair of exponentials $\bang$ and $\quest$, \sell\ may
contain as many \emph{subexponentials}~\cite{danos93kgc,nigam09ppdp} as needed.


Figure~\ref{fig:ll} presents the introduction rules of the 
fragment of linear logic that will be used in order to 
build soft constraint system ($\otimes,\exists,\one,\forall, \top$) and to give meaning to processes  ($\with,\lolli$). 
Note that formulas are not always allowed to contract and weaken: this is 
controlled in linear logic by the use of the {\em exponentials} $\bang$
and $\quest$. In \sell, this control is finer since it is possible to specify which subexponentials behave classically or not. 

\begin{figure}
\[
 \infer[\tensor_L]{\Gamma, F \tensor H \lra G}
{\Gamma, F, H \lra G} 
\quad 
\infer[\tensor_R]{\Gamma_1, \Gamma_2 \lra F \tensor H}
{\Gamma_1 \lra F & \Gamma_2 \lra H}
\quad
\infer[{\one}_L]{\Gamma, \one \lra G}
{\Gamma \lra G}
\quad 
\infer[{\one}_R]{ \lra \one}{}
\quad 
\infer[\top_R]{ \Gamma\lra \top}{}
\]

\resizebox{\textwidth}{!}{
$
\infer[\lolli_L]{\Gamma_1, \Gamma_2, F \lolli H \lra G}
{\Gamma_1 \lra F & \Gamma_2, H \lra G}
\quad 
\infer[\lolli_R]{\Gamma \lra F \lolli H}{\Gamma, F \lra H}
\quad 
\infer[\with_{L_i}]{\Gamma, F_1 \with F_2 \lra G}
{\Gamma, F_i\lra G} 
\quad 
\infer[\with_R]{\Gamma \lra F \with H}
{\Gamma \lra F & \Gamma \lra H}
$}

\[
\infer[\exists_L]{\Gamma, \exists x. F \lra G}
{\Gamma, F[e/x] \lra G}
\quad 
\infer[\exists_R]{\Gamma \lra \exists x.G}
{\Gamma \lra G[t/x]}
\quad
\infer[\forall_L]{\Gamma, \forall x. F\lra G}
{\Gamma, F[t/x] \lra G}
\quad
\infer[\forall_R]{\Gamma \lra \forall x. G}
{\Gamma \lra G[e/x]}
\]
%
\caption{A fragment of the LL  introduction rules. Here $e$ is a fresh variable and $t$ is a term.}
\vspace{-2mm}
\label{fig:ll}
\end{figure}

Formally, a \sell\ system is specified by a \emph{subexponential signature} $\Sigma = \tup{I, \preceq,U}$, 
where $I$ is a set of labels, $U \subseteq I$  specifying which
subexponentials allow both weakening and contraction, and $\preceq$ 
is a pre-order among the elements of $I$.
 We shall use $a,b,\ldots$
 to range over elements in $I$ and we will assume that $\preceq$
is upwardly closed with respect to $U$, \ie, if $a \in U$ and $a \preceq
b$, then $b \in U$. For a given such subexponential signature, $\sell_\Sigma$ is 
the system obtained by substituting the linear logic exponential $\bang$ by the subexponential $\nbang{a}$ for each $a\in I$, and by adding to the rules in Figure~\ref{fig:ll} the following inference rules:

\noindent - 
 for each $a \in I$ (dereliction and the promotion rules):
  \[
  \infer[\nbang{a}_L]{\Gamma, \nbang{a} F \lra G }{\Gamma,F \lra G }
  \qquad 
  \infer[\nbang{a}_R \textrm{, provided $a \preceq a_i$ for $1 \leq i \leq n$.}]{\nbang{a_1} F_1, \ldots, \nbang{a_n} F_n \lra \nbang{a} F }
  {\nbang{a_1} F_1, \ldots, \nbang{a_n} F_n \lra  F }
 \]
\noindent -  for each $b \in U$ (structural rules):
 \[
  \infer[W]{\Gamma, \nbang{b} F \lra G }{\Gamma \lra G }
  \qquad
  \infer[C]{\Gamma, \nbang{b} F \lra G }{\Gamma, \nbang{b} F, \nbang{b} F \lra G }
 \]
In this paper we will not use the $\nquest{a}$ subexponential, since the specifications will be within the {\em minimal} setting of \sell. We would like to stress out that this choice do not affect the expressiveness of the framework, as pointed out in~\cite{DBLP:conf/csl/Chaudhuri10}.
Observe that provability is preserved {\em downwards} i.e. the sequent $\Gamma \lra \nbang{a} P$ is provable in $\sell_\Sigma$, then 
so is the sequent $\Gamma \lra \nbang{b} P$ for all $b\preceq a$. We shall elide the signature $\Sigma$ whenever it is not important or clear from the context.




Subexponentials greatly increase
the expressiveness of the system when compared to linear logic. The key difference  is that  while linear logic has only seven logically
distinct prefixes of $\bang$ and $\quest$ (e.g., $\bang F$, $\bang\quest\ F$, $\quest\bang F$, etc) \cite{danos93kgc}, \sell\ allows for an
unbounded number of such prefixes (e.g., $\nbang{a}{}\nquest{b}F,\nbang{b}{}\nquest{a}F$, etc). In fact,  in~\cite{DBLP:conf/concur/NigamOP13}, we showed that  by using
different prefixes 
it is possible  to interpret
subexponentials in more creative ways, such as temporal
units  or spatial and epistemic modalities.

  $\sell$  enjoys good proof theoretic properties. For instance, \cite{danos93kgc} proved that 
  $\sell$ admits cut-elimination.
Moreover, \cite{nigam09ppdp} proposed a sound and complete focused proof system~\cite{andreoli92jlc} for 
$\sell$. 
In this work, however,  we will use an unfocused version of $\sell$, since extending focusing to the  \sells\ system  (see Section~\ref{sells}) is a non trivial task.

\vspace{-2mm}
\subsection{C-semiring as Subexponentials Signatures}\label{sec-semiring-sells}
In~\cite{DBLP:conf/concur/NigamOP13}
we studied the logical meaning
of \ccp\ processes    as \sell\ formulas. For that, we assumed that  the underlying constraint system had a logical structure and  we required  simple pre-orders as subexponential signatures. Here   we go in the opposite direction: assuming 
that \ccp\ processes can be endowed with a logical meaning, we 
propose a logical framework for building {\em soft} constraints, thus recovering the logical reading of Soft \ccp\ systems. 
 This requires a more involving  algebraic 
structure in the subexponential signature, as follows.
%
%
\begin{definition}[Soft Constraint System (SCS)]\label{def:soft-cons}
Let $S = \langle \cA,+_\cA,\times_\cA, \botA,\topA \rangle$ be a c-semiring with $\leqA$ the order induced by $+_\cA$;  $\cP$ be a first order signature
;   $\Sigma = \tup{\cA, \leqA,\cA}$  be a subexponential signature; and  $\cC$ be a set of  \sell\
formulas built from the syntax:
 \[
 \begin{array}{lll}
{\mathbf C}& ::= & \one \mid   C \otimes C\mid \exists x.(C) \mid \nbang{a}A \mid \nbang{a}(\nbang{a}A_1 \otimes \cdots \otimes \nbang{a}{A_n})
\end{array}
 \]
 where $a\in\cA$ and $A, A_i$ are atomic formulas  (\ie,  predicate symbols in $\cP$ applied to terms).
  Elements in $\cC$, with typical elements $c,d$, are called constraints. Let   $\Delta=\{\delta_1,...,\delta_n\}$   be a (possibly empty) set of non-logical axioms of the form $\forall\vx_i. (c _i\limp d_i)$ where all free variables in  $c_i$ and $d_i$ are in $\vx_i$. 
A soft constraint system SCS
is a structure 
 $\langle\cA,  \cC, \entails  \rangle$ 
 where   $d \entails c$ iff the sequent 
 $\nbang{\topA}{\delta_1},..., \nbang{\topA}{\delta_n}, d \lra c$ is provable in SELL. 
\end{definition}

We shall call \emph{pre-constraints} 
formulas of the shape $A_1 \otimes \cdots\otimes A_n$ or  an atom $A$. 
As usual in the specification of constraint systems as formulas in a given logic, 
the previous definition built constraints from the
the empty store  ($\one$); conjunction of constraints ($\otimes$); and existential quantification of constraints. In our case, additionally, a constraint can be   a formula $F$ of the form $\nbang{a}{A}$ or  $\nbang{a}({\nbang{a}A_1 \otimes \cdots \otimes \nbang{a}{A_n}})$ where 
$A, A_i$ are atomic formulas. Roughly, $F$ means that the pre-constraint $A$ (or $A_1 \otimes \cdots \otimes A_n$) was added to the store with an upper bound preference 
degree $a$. Note that $a$ is a c-semiring value and, according to the previous definition, it is a subexponential. Moreover, due to the signature $\Sigma$, all the subexponentials are unbounded which means that soft constraints cannot be removed from the store. In what follows, we shall 
write $\softC{A}{a}$  instead of $\nbang{a}{A}$;   $\softC{A_1 \otimes \cdots \otimes A_n}{a}$
instead of $\nbang{a}({\nbang{a}A_1 \otimes \cdots \otimes \nbang{a}{A_n}})$; and   $\softCT{F}$ instead of $\softC{F}{\top_{\cA}}$. 

Now  we shall show that  our construction is indeed an instance of the general definition of CS 
as cylindric algebras in  \cite{saraswat91popl,BoerPP95}. This guarantees that all the machinery developed for \ccp\ calculi can be used also when considering  programs with the  SCS in Definition \ref{def:soft-cons}. 
Roughly, a Cylindric Constraint System is a structure  $\langle \cC,\leq,\sqcup,\one,\zero,{\it Var}, \Exists,
D \rangle$ where $\cC$ is a set of tokens (constraints); $c\leq d$ iff $d \entails c$; $\Exists$ is a cylindrification operator that models hiding of variables; and $D\subseteq \cC$ is the set of diagonal elements of the form $d_{xy}$  that can be thought of as the equality $x=y$. In the  \ref{sec:cylindric}  the reader may find the complete definition of these systems and the proof of the theorem below. We note 
diagonal elements (and axioms in Definition \ref{def:soft-cons}) 
are marked with the   subexponential $\topA$.  Hence, the sequent $\softCT{d_{xy}} \lra \softC{d_{xy}}{a}$ is provable for any   $a \in \cA$. 
Intuitively, this means that axioms and diagonal elements are \emph{available} (and can be used) under any preference level. 

%

\begin{theorem}[Constraint System]\label{thm:constraints}
Let  $\mathbb{C} = \langle\cA,  \cC, \entails \rangle$ be as in Definition~\ref{def:soft-cons}. Then, the structure $\langle \cC,\leq,\otimes,\one,\zero, Var, \exists, D \rangle $ is   a cylindric constraint system where $D = \{\nbang{\top_{\cA}}(d_{xy}) \mid  x,y\in Var\}$  and $c\leq d$ iff $d\entails c$. 
\end{theorem}

\vspace{-3mm}
\subsection{Logical Reading of Processes}
\label{sec:logical-reading}

In ~\cite{DBLP:conf/concur/NigamOP13} we extended the results in  \cite{fages01ic} and we showed that \ccp\ processes have a strong connection with ILL: operational steps matches 
{\em exactly} focused logical steps~\cite{nigam09ppdp}. We also showed that such characterization extends to  various \ccp\ calculi like epistemic,  spatial and timed systems. 

%
%

Unlike the results in~\cite{DBLP:conf/concur/NigamOP13}, 
  the encoding here considers 
  non-determinism and we 
  do not need the extension of \sell\ with families
  or quantification over subexponentials,
as the  systems of soft constraints do not mention nested modalities. It seems possible, however, to include these modalities to obtain Soft-\ccp\ systems that mention spatial or temporal
modalities (see Section \ref{sec:concluding}).



Assume a SCS 
and let $\Sigma'=
 \tup{\cA \cup \{\pfamily,\dfamily, \ufamily\}, \preceq, \cA \cup \{\ufamily\}}
$ be a subexponential signature  where, for any $a,b\in\cA$, $a\preceq b$ iff $a\leq_{\cA} b$, and $a,\pfamily,\dfamily,\ufamily$ are unrelated wrt  $\preceq$.
Observe that $a,\ufamily  \in U$ while $\pfamily,\dfamily\notin U$. 
Intuitively, the subexponential $\pfamily$ is used
to mark processes; 
$\ufamily$ marks process definitions; and
 $\dfamily$ marks calls  $p(\vx)$ whose
definition may be unfolded. 
We will build the subexponential signature $\Sigma$  from $\Sigma'$, as the
completion of $\Sigma'$ to a c-semiring. This 
is easily achieved by adding two distinguished elements: $\bottom_c,\top_c$ such that $\bottom_c\leq_{\cA}\botA$;
$\bottom_c\leq_{\cA} \pfamily,\dfamily,\ufamily$;
$\topA\leq_{\cA}\top_c$; and 
$\pfamily,\dfamily,\ufamily \leq_{\cA} \top_c $.  Then, for example, $\pfamily\times_{\cA} a=\bottom_c$ and $\pfamily+_{\cA} a=\top_c$ for any $a\in\cA$. 

Now we show how processes can be given a logical meaning as formulas in \sell. 

\begin{definition}[Encoding of processes, non-logical axioms and process definitions]\label{def:se-csystem}
For any process $P$,  $\pEnc{P}{}$ is defined recursively as:
\[
\begin{array}{l@{\qquad}l}
\bullet~\pEnc{\tellp{c}}{} = \nbang{\pfamily} c &
\bullet~\pEnc{\sum\limits_{i\in I}\whenp{c_i}{P_i}}{} = \nbang{\pfamily} \bigwith_{i\in I} (c_i\limp \pEnc{P_i}{})\\
\bullet~\pEnc{\localp{\vx}{P}}{} =  \nbang{\pfamily}( \exists \vx. \pEnc{P}{}) &
\bullet~\pEnc{P_1,...,P_n}{} =\pEnc{ P_1}{}   \ot ... \ot   \pEnc{P_n}{} \\
\bullet~\pEnc{p(\vx)}{} =\nbang{\dfamily} p(\vx) & 
\end{array}
\]
Recall that non-logical axioms are encoded as formulas of the form $\nbang{\top_{\cA}}(\forall\vx (d \limp c))$ (see Def. \ref{def:soft-cons}). A 
process definition of the form $p(\vx)\defsymboldelta P$
is encoded as $
\nbang{\ufamily}[ \forall \vx.(\nbang{\dfamily}p(\vx) \limp \pEnc{P}{})].
$
\end{definition}  
We can now state the adequacy theorem, where
$\enc{\Psi}$ represents the set of \sell\ formulas encoding
the set of 
 process definitions $\Psi$. The proof is in  \ref{sec:adequacy}. 

\begin{theorem}[Adequacy]
\label{thm:adeq-basic}
Let $P$ be a process,
$(\cA,\cC,\entails)$ be  a SCS 
with a (possible empty) set of non-logical axioms $\Delta$ and   $\Psi$ be a set of process definitions.
Then $\Barb{P}{c}$ iff $ \Delta, 
  \enc{\Psi}, \pEnc{P}{} \lra
c \otimes \top$.
\end{theorem}



\section{Computing with Soft Constraints}
\label{sec:computing}
\vspace{-2mm}

In this section we show how to compute with soft constraints. We distinguish two classes of SCS  according to the underlying c-semiring:  idempotent and non-idempotent. 
First, from the pre-order induced by c-semiring (Definition~\ref{def:soft-cons}), we can rephrase the
side-condition of \sell's promotion rule for SCS  as follows:
\begin{equation}\label{eq:glb}
  \infer[\nbang{a}_R \textrm{, provided $a \leq_{\cA} glb(a_1,...,a_n)$}]{\nbang{a_1} F_1, \ldots, \nbang{a_n} F_n \lra \nbang{a} F }
  {\nbang{a_1} F_1, \ldots, \nbang{a_n} F_n \lra  F }
\end{equation}

\vspace{-4mm}
\subsection{Idempotent soft constraints}

It turns out that, in an idempotent c-semiring, $a\times_{\cA} b = glb(a,b)$. Hence the side-condition in (\ref{eq:glb})
is equivalent to 
\begin{equation}\label{eq:idem}
a \leq_{\cA} a_1 \times_{\cA}...\times_{\cA} a_n
\end{equation} 

For illustrating better how the promotion rule is
used in idempotent systems, consider
the  fuzzy c-semiring $S_F = \langle [0,1], max, min, 0,1\rangle$
and its corresponding SCS as in Definition \ref{def:soft-cons}. Let $c,d$ be pre-constraints and consider    $T = P \parallel Q \parallel R \parallel S$ where:
\[
\begin{array}{l l l }
P = \tellp{\softC{c}{0.7}} \parallel \tellp{\softC{d}{0.2}} & \qquad & 
Q = \whenp{\softC{c }{0.3}}{Q'} \\
R = \whenp{\softC{c\otimes d}{0.5}}{R'} & \qquad & 
S = \whenp{\softC{c\otimes d}{0.2}}{S'} \\
\end{array}
\]

From the initial store $\one$, we observe the following transitions:\\

\resizebox{\textwidth}{!}
{
$
\begin{array}{lll}
\conf{\emptyset}{T}{\one}~  \redi^* ~ \conf{\emptyset}{Q\ \parallel R \parallel S\ }{\softC{c}{0.7} \otimes \softC{d}{0.2}} ~ \redi^* ~ \conf{\emptyset}{Q'\parallel R \parallel S'}{\softC{c}{0.7} \otimes \softC{d}{0.2}}
\end{array}
$
}
\\
The ask   $Q$ can proceed since the sequent $\softC{c}{0.7} \lra \soft{c}{0.3} $ is provable. 
Furthermore, since the sequent 
$\softC{c}{0.7}, \softC{d}{0.2} \lra \softC{c \otimes d}{0.2}$ is also provable, $S$ can evolve into $S'$. Finally,   $R$ remains blocked since  the sequent $\softC{c}{0.7}, \softC{d}{0.2} \lra \softC{c \otimes d}{0.5}$ is not provable:
 introducing $\nbang{0.5}$ on the right implies weakening the formula $\softC{d}{0.2}$ on the left.  That is, the process $P$ adds the information that $c$ (resp. $d$) is \emph{preferred} with a level of $0.7$ (resp. $0.2$). 
Hence the pre-constraint $c\otimes d$ can be deduced only with a preference level less or equal to $0.2$. 

\subsection{Non-idempotent soft constraints}\label{sells}

 It is well known that some of the interesting properties of the c-semiring framework for constraint programming do not  hold for non-idempotent c-semirings (see Section \ref{sec:r-work}). 
In our framework, if   $\times_{\cA}$ is not idempotent then  it may be the case that  $a\times_{\cA} b <_{\cA} glb(a,b)$; hence
the side conditions in (\ref{eq:glb}) and (\ref{eq:idem}) are no longer 
equivalent and therefore the promotion rule in (\ref{eq:glb}) does not seem to be adequate anymore.

For an example, let $S_P = \langle[0,1], max, \times, 0 ,1 \rangle$ be the probabilistic c-semiring and  $T$ be the process as above. We notice that under this SCS, the sequent $\softC{c}{0.7}, \softC{d}{0.2} \lra \softC{c \otimes d}{0.2}$ is provable (as in the case of the Fuzzy c-semiring) and then, the process $S$ can proceed.
This does not fit to our intuition that 
from $\softC{c}{0.7} \otimes \softC{d}{0.2}$
we can only entail $c\otimes d$ with a probability less or equal to $0.14$. This undesired behavior comes  with no surprise since the provability relation takes into account the ordering $\leq_\cA$  induced by the $+_\cA$ operator  but it does not   ``combine'' information  with the $\times_\cA$ operator. 
%
%


Fortunately, it is possible to redefine the promotion rule  
in order to specify the ``combination'' of c-semiring values when non-idempotent c-semirings are considered. 
We define the system \sells\ from \sell, replacing the side condition of the promotion rule. 
\begin{definition}[\sells\ system]\label{def:sells}
Let $\Sigma$ be a subexponential signature as in Definition~\ref{def:soft-cons}. 
The  $\sells_\Sigma$ system shares with \sell\ all the rules but the  promotion rule, which is defined as
\[
  \infer[\nbang{a}_{R_S} \textrm{, provided $a\leq_{\cA} a_1\times_{\cA}\ldots\times_{\cA} a_n$}]
  {\nbang{a_1} F_1, \ldots, \nbang{a_n} F_n \lra \nbang{a} F }
  {\nbang{a_1} F_1, \ldots, \nbang{a_n} F_n \lra  F }
\]
We shall write \sells\ instead of $\sells_\Sigma$ when $\Sigma$ can be inferred by the context.
\end{definition}
Note that, for an idempotent c-semiring, this condition is the same as in the \sell\ system since $ a_1 \times_{\cA} \cdots \times_{\cA} a_n = glb(a_1,  \cdots, a_n)$. In the case of non-idempotent c-semirings, though, this condition is \emph{stronger} since $a_1\times_{\cA} \cdots \times_{\cA} a_n \leq glb(a_1,\cdots,a_n)$. 
The new rule 
is not at all ad-hoc: while $\sells$ is a smooth extension of ILL, it is a  closed
subsystem of \sell\ which is strict when non-idempotent c-semirings are considered.  
Hence $\sells$ inherits all \sell\ good properties, such as cut elimination (see the proof in~\ref{sec:cut}).
\begin{theorem}
\label{thm:sells-cut}
 $\sells$ admits cut-elimination.
\end{theorem}
We note that Theorem~\ref{thm:adeq-basic} is also valid for the
non-idempotent case as shown in Appendix~\ref{seq:adSELLS}.
Observe also that the rule $\nbang{a}_{R_S}$ above has a strong 
synchronous flavor: not only  it inherits the synchronous behavior of the bang, but
it also 
introduces a strong non-determinism on 
choosing the formulas on the left-hand-side of the sequent marked with exponentials $a_1, \ldots, a_n$. 

Finally, notice that, in \sells, the sequent $\softC{c}{0.7}, \softC{d}{0.2} \lra \softC{c \otimes d}{0.2}$ is no longer provable while $\softC{c}{0.7}, \softC{d}{0.2} \lra \softC{c \otimes d}{a}$ is provable whenever $a\leq 0.14$, as desired. 
This finer way of controlling subexponentials on the left side of sequents allows considering
more interesting spaces as signatures, and it
opens the possibility of specifying more challenging
computational systems. 



\vspace{-2mm}
\subsection{Monotonicity and level of preferences}
Let us now explain how the Soft-\ccp\ language here proposed adheres to the 
elegant properties of its predecessors. 
In \ccp\ languages, the store grows monotonically, \ie, one can easily verify by induction on the structure of $P$ that if $\conf{X}{\Gamma}{c} \redi^* \conf{X'}{\Gamma'}{c'}$ then $\exists X'. c' \entails \exists X. c$. 
In c-semiring based constraints, when two constraints are combined,  one gets a lower value of the c-semiring. In the case of constraint solving  and soft concurrent constraint programming as in \cite{DBLP:journals/tocl/BistarelliMR06}, 
this can be understood as the fact that having more constraints implies that it is ``more difficult'' to satisfy all of them. Hence we have: (i) more constraints imply a stronger store and then, more information can be deduced from it; and (ii) more constraints imply a lower level of preference in the semiring. How should we interpret these somehow contradictory ideas?

 This problem was already addressed in \cite{DBLP:journals/tocl/BistarelliMR06} where the entailment relation (that is only defined for idempotent c-semirings --see Section \ref{sec:r-work}) is defined as the inverse of the ordering of the semiring. Roughly speaking, $C$ entails $c$ iff $\bigotimes C \sqsubseteq c$ where 
$\bigotimes C$ denotes the combination ($\times_{\cA}$) of constraints in the set $C$ and $\sqsubseteq$ is the ordering induced by $\leq_{\cA}$ on constraints. 

Now let us explain how (i) and (ii) above coexists in our framework.  We note first that the sequent 
$\softC{c_1}{a_1}, \ldots, \softC{c_n}{a_n}, \softC{c}{a} \lra \softC{c}{b}$ 
is provable for any $b\leq_{\cA} a$ 
 This means that, if an agent adds the constraint $c$ with level of preference $a$, then it is possible to deduce $c$ with a preference level less or equal to $a$. That is, the store grows monotonically. Now consider the store 
$\softC{c_1}{a_1} \otimes \cdots \otimes \softC{c_n}{a_n} \otimes \softC{c}{a} \otimes \softC{c}{b}$
where $a<_{\cA}b$. In this case, $c$ can be deduced with a preference level less or equal to $b$. This also matches the monotonic behavior we want in the store: if $\softC{c}{b}$ is added first, then the agent adding $\softC{c}{a}$ is just adding ``irrelevant'' information to the store that can be weakened when needed; on the other hand, if $\softC{c}{a}$ is added first, then, adding $\softC{c}{b}$ means that $c$ is believed with a greater level of preference and the store becomes stronger. 
Consider the stores $d_1= \softC{c}{a}\otimes \softC{d}{b}$ and $d_2= \softC{c\otimes d}{b}$
where $a <_{\cA} b$. If $e \leq_{\cA} a \times_{\cA} b$, 
it is clear that $d_1 \entails \softC{c\otimes d}{e}$. Moreover,  $d_2 \entails d_1$. This shows that 
believing both $c$ and $d$ with a given preference level $b$ (i.e., $[c\otimes d]_b$) is stronger than believing  $c$ with a 
preference level $a\leq_{\cA} b$. Note that the sequent $d_2 \lra d_1$ is provable 
only 
because all atoms in the constraint system are classical -- in our example,  $d_2 = \nbang{b}(\nbang{b}{c}\otimes \nbang{b}{d})$. Finally, the store is idempotent as in \ccp\ ($c \sqcup c \equivC c$). To see that, 
notice that  $\softC{c}{a} \otimes \softC{c}{a}  \equiv \softC{c}{a} $ (regardless the idempotency of
$\times_{\cA}$).

\vspace{-2mm}
\section{Concluding Remark}
\label{sec:concluding}
\vspace{-2mm}

We have established a tight connection between Soft-\ccp\ systems and linear logic proof systems. In particular, we investigated the 
use of subexponentials in linear logic with more involved pre-orders as logical foundations for soft constraints. 
Moreover, we have also proposed a novel proof
system, \sells, giving a logical meaning to soft constraints based on non-idempotent semirings, such as probabilistic and weighted 
soft constraints.

\paragraph{\bf Related Work}\label{sec:r-work}

In \cite{DBLP:journals/tocl/BistarelliMR06} the first \ccp\ language featuring soft constraints was proposed. There, c-semiring based constraints, seen as functions mapping variable assignments into c-semiring values, are lifted to a higher-order semiring where constraints can be combined and compared. In such formalization, an entailment relation \`a la Saraswat \cite{saraswat91popl} can be defined only if the $\times_{\cA}$ operator is idempotent (see \cite[Def. 3.8, Th. 3.9]{DBLP:journals/tocl/BistarelliMR06}). In 
particular, given a set of constraints $C$, 
if $\times_{\cA}$ is non-idempotent,  $C \entails d$ does not imply that  $C \sqcup d \equiv C$. Note that
in our case, if $C \lra d$ then the equivalence 
$ (\bigotimes C \otimes d) \equiv (\bigotimes C)$ is provable   (regardless the idempotency of $\times_\cA$). Hence, our logical characterization of soft constraints as formulas in \sell\ follows closely the idea of monotonic store in \ccp. 

The language proposed in 
\cite{DBLP:journals/tocl/BistarelliMR06} allows agents to 
be \emph{guarded} by a semiring value $a\in \cA$. Hence, an agent performs an action only if  the resulting store is not \emph{weaker} than the cut level $a$.  For instance,  $\tellp{c} \lra^a P$ adds $c$ to the store and then executes $P$ if $c$ in conjunction with the current store has a level of preference greater than $a$. We could also add to our language such kind of constructs 
by modifying accordingly the SOS in order to handling $a$-guarded constructs. Nevertheless, one should be careful since the logical meaning of processes is lost (Theorem \ref{thm:adeq-basic}).  The main reason is that such constructs do not have a proof theoretically meaning: it is necessary to check the consistency of the system first, to latter add a formula to the context.

The work in  \cite{DBLP:conf/coordination/BistarelliGMS08} combines the notion of time in $\texttt{tccp}$ \cite{DBLP:journals/iandc/BoerGM00} with soft constraints. Due to Theorem \ref{thm:constraints}, a similar extension can be also done with our framework by plugin into \tcc\ \cite{DBLP:journals/jsc/SaraswatJG96}
or \texttt{tccp}    the soft constraint system in Definition \ref{def:soft-cons}.
 Moreover, due to the logic inspiration of the constraint system proposed here, it is possible to show also that timed processes manipulating soft constraints can be declaratively characterized as formulas in SELL (\cite{DBLP:conf/concur/NigamOP13}). 

A model-based (semantic) characterization of soft constraints based on c-semirings is given in \cite{DBLP:journals/heuristics/Wilson06}. To the best of our knowledge, ours is the first proof-theoretic characterization of such systems.  However, the use of more involved orders
for subexponentials is not completely new. They were used recently in different contexts, such as in Bounded Linear Logic~\cite{ghica.unp} and 
in programming languages~\cite{DBLP:conf/esop/BrunelGMZ14}. 

\vspace{-4mm}
\paragraph{\bf Future Work}
We can foresee several research directions from this work. From the point of view of proof theory, the proof system \sells\ is 
novel. We are currently investigating a focused proof system for it, which seems to be a non trivial task: the key problem is 
how to handle contraction of formulas. In fact, when contracting a formula one is no longer able to prove formulas marked with some subexponential
bang. This is different from \sell. It seems possible, however, to use the fact that subexponentials are unbounded to come up with 
a sensible focused proof system for (fragments of) \sells. 

The definition we gave for soft constraint systems is general enough to be used in different \ccp\ idioms. In particular, it is possible to define systems with spatial information where agents can \emph{believe} the same information with different levels of preferences. 
Theorem \ref{thm:adeq-basic} along with the logical characterization of spatial \ccp\ in \cite{DBLP:conf/concur/NigamOP13} may allow us to prove correct such approach. We also foresee systems where agents can update their preferences. For that, we shall need to use quantifiers over subexponentials as defined in \cite{DBLP:conf/concur/NigamOP13}. Finally, it seems that we can define our subexponentials to be linear in order to have declaratively some forms of \emph{retraction} of soft constraints.  

\paragraph{Acknowledgments} We thank Francesco Santini 
for helpful discussions. Nigam was supported by CNPq and Pimentel was supported by CNPq and CAPES.  The work of Olarte has been (partially) supported by Colciencias (Colombia), CNPq and by Digiteo and DGAR (\'Ecole Polytechnique) funds for visitors.





\newpage
\appendix
\label{sec:appendix}

\section{Adequacy Theorem}\label{sec:adequacy}
In this section we will discuss the adequacy theorem. 
We will start by proving Theorem~\ref{thm:adeq-basic} 
for the case where the framework used for the specification is  \sell\ (i,e, the underlying constraint system is built from an idempotent c-semiring). Later, in Section \ref{seq:adSELLS}, we extend this result for the \sells\ case for non-idempotent c-semirings. 
\subsection{Adequacy using \sell}\label{seq:adSELL}
Since \sell\ admits a {\em focused} system \cite{andreoli92jlc},
we can use here the same machinery developed in~\cite{DBLP:conf/concur/NigamOP13}. 

 First of all, notice that, by using simple logical equivalences (such as 
 moving the existential outwards), we can rewrite 
 the constraints to the  following shape:
 $$ c= \exists \vx.(\softC{pc_1}{a_1}\tensor \cdots \tensor \softC{pc_n}{a_n})
$$
 where $\softC{pc_1}{a_1}, \ldots, \softC{pc_n}{a_n}$ are all of the form 
 $\nbang{a_i}(\nbang{a_i} A_1 \tensor  \cdots \tensor \nbang{a_{i}} A_{mi})$ or of the form $\nbang{a_i} A$.  
 Observe that the formula above is composed only by  positive
 formulas. Thus, from the focusing discipline, whenever such a formula appears in the left-hand-side, it is 
 decomposed as illustrated by the following derivation:
 \[
  \infer[p \times \exists_L]{ \Delta, \exists \vx.(\softC{pc_1}{a_1}\tensor \cdots \tensor \softC{pc_n}{a_n}) \lra \Rscr}{
  \infer[n-1 \times \tensor_L]{ \Delta, \softC{pc_1}{a_1}\tensor \cdots \tensor \softC{pc_n}{a_n} \lra \Rscr}
  { \Delta, \softC{pc_1}{a_1}, \ldots , \softC{pc_n}{a_n} \lra \Rscr}
  }
 \]
 Next, the constraints $\softC{pc_1}{a_1}, \ldots , \softC{pc_n}{a_n}$ 
 appearing in the premise of this derivation are moved to the contexts $a_1, \ldots, a_n$, respectively. This is all done in a negative phase.
That is, focusing on $\pEnc{\tellp{c}}{}$ corresponds exactly to the operational semantics of tells: the pre-constraints in $c$ are added
to the constraint store, creating fresh names in the process.

On the other hand, if such a constraint $c$ is focused on the right, the derivation will have the shape
$$\infer=[p \times \exists_R]{\Delta\rfoc{\exists \vx.(\softC{pc_1}{a_1}\tensor \cdots \tensor \softC{pc_n}{a_n})}}
{\infer=[n-1 \times \tensor]{\Delta\rfoc{\softC{pc_1}{a_1} \tensor \cdots \tensor \softC{pc_n}{a_n}}}
{\infer{\Delta_1\rfoc{\softC{pc_1}{a_1}}}
{\deduce{{\Delta_1}_{\leq a_1}\lra pc_1 }{}} &
 \cdots 
 \quad
\infer{\Delta_n\rfoc{\softC{pc_n}{a_n}}}
{\deduce{{\Delta_n}_{\leq a_n}\lra pc_n}{}}}}
$$
where $\Delta\rfoc{c}$ represents a sequent 
with left context $\Delta$ and
focused 
on the right-hand side formula $c$. Here
${\Delta_i}_{\leq a_i}$ contains the elements of 
$\Delta_i$ whose contexts are marked with subexponentials greater or equal to $a_i$. Since $a_i$ 
and $\ufamily,\pfamily,\dfamily$ are not related, ${\Delta_i}_{\leq a_i}$ will have
only pre-constrains and non-logical axioms.
This means 
that focusing on $\pEnc{\whenp{c}{P}}{} = \nbang{\pfamily}  (c\limp \pEnc{P}{})$
corresponds to proving $c$ only from pre-constraints and non-logical axioms and moving all the other resources to proving 
$\pEnc{P}{}$.

Continuing this exercise, we can go case by case and prove that, indeed, one focus step corresponds to one operational step, hence proving Theorem~\ref{thm:adeq-basic}
with the highest level of adequacy (on derivations).

\subsection{Adequacy using \sells}\label{seq:adSELLS}
The ideas above cannot be used in order to show that the adequacy theorem also holds for \sells. The reason is that it is not trivial how to define a focused system to \sells.
Thus 
we will show  that, in the proof of constraints, no
encoded processes,
procedure calls or procedure definitions are
   used. 
This is due to the fact that $\ufamily,\pfamily,\dfamily$ are unrelated, and $\pfamily,\dfamily$ are linear.

\begin{lemma}
\label{lem:notprovable}
Assume the subexponential signature $\Sigma$ used to build Soft-\ccp. Let
$\Delta \cup \{p,c\}$ be a set of formulas, where: $\Delta$ contains the encoding of non-logical axioms and constraints;  $c$ is a constraint  and $p$ is the encoding  of a process or of a procedure call. Let  $b$ be the subexponential 
$\pfamily$ or $\dfamily$.
Then the sequents $\Delta, \nbang{b} p \lra  c$  and
$\Delta, p \lra c$ are not provable 
in $\sells_\Sigma$.
\end{lemma}
\begin{proof}
The proof is by contradiction. Assume that the sequent 
$\Delta, \nbang{b} p \lra c$ (resp. 
$\Delta, p \lra  c$)
is provable and consider a proof $\pi$ of it with
smallest height. The last  rule applied in $\pi$
cannot be an initial rule, because $\nbang{b} p$ 
(resp. $p$) is linear. One possible
action is to derelict the formula $\nbang{b} p$ obtaining the 
sequent $\Delta, p \lra c$, which reduces the two 
cases to one.
Another possibility  would be applying some non logical axiom $\nbang{\top_{\cA}}(\forall\vx (d \limp e))$ in $\Delta$.
But since $d,e$ are constraints, this 
will lead to a premise with the formula $ \nbang{b} p$ (resp. $p$) in the context. 
Moreover, introducing the formula $c$
is either not possible: when $c$ is of the form $\nbang{a} pc$, 
$b$ is unrelated to $a$ (resp. the linear formula $p$ is in the context);  
or when possible, that is, when $c$'s main connective is an $\exists$ or a 
$\tensor$, then $ \nbang{b} p$ (resp. $p$) is in the context of one of the 
premises. 
Finally, we can introduce the formula $p$ if it is the encoding of a process, such as 
an ask. But again one of the  resulting premises will again contain a formula of the form $\nbang{b} p'$ in the context,
where $p'$ is the encoding of a process. 
Thus there is no such minimal proof.
\end{proof}

\begin{lemma}
Assume the subexponential signature $\Sigma$ used to build Soft-\ccp. Let
$\Delta \cup \{f,c\}$ be a set of formulas, where $\Delta$ contains
the encoding of logical axioms and constraints; $c$ is a constraint, and $\nbang{\ufamily} f$ is the encoding of a 
process definition $p(\vx)\defsymboldelta P$.
Then the sequent $\Delta, \nbang{\ufamily} f \lra c$ is provable in $\sells_\Sigma$
if and only if $\Delta \lra c$ is provable.
\end{lemma}
\begin{proof}
The $(\Leftarrow)$ direction is straightforward as one only needs to weaken  
$\nbang{\ufamily} f$. 

The $(\Rightarrow)$ direction is as follows. The only way to prove the sequent 
$\Delta, \nbang{\ufamily} f \lra c$ is by weakening $\nbang{\ufamily} f$. 
As in the proof of Lemma~\ref{lem:notprovable}, either we cannot introduce $c$
or when it is introduced the formula $\nbang{\ufamily} f$ still appears in the context of the premise.
Moreover, we cannot derelict $\nbang{\ufamily} f$, because the resulting sequent would
contain a linear formula and using the same reasoning in Lemma~\ref{lem:notprovable}
we can show that this resulting sequent is not provable. Contracting $\nbang{\ufamily} f$
also does not help in the proof, as the new occurrence of $\nbang{\ufamily} f$ would 
also need to be weakened.
\end{proof}

Hence even without using focusing in order to control the flow of the proof, we have a neat way of controlling
its shape, using the subexponential structure and linearity.

\section{Cut-elimination for \sells}\label{sec:cut}
We prove now Theorem~\ref{thm:sells-cut}. 
We shall omit the subindex ``$\cA$'' in $\times_{\cA}$ and $+_{\cA}$ since in this context it is clear that   $\times$ and $+$ refer to the operands of the c-semiring. 

We start by proving the following result, which 
is a substitution lemma for $\preceq$.  
\begin{lemma}
\label{lemma:subst}
 Let $\Sigma$ be a subexponential signature constructed
 on a c-semiring. Then if $b \preceq a\times c$ and 
 $a \preceq d$, then 
 $b \preceq d \times c$.
\end{lemma}

\begin{proof}
Let's assume that $b \preceq a\times c$  and $a \preceq d$. We prove $b \preceq c\times d$.
Recall that $x \preceq y$ if $x + y = y$ (by definition). Then
$b \preceq a\times c$ iff  $b+a\times c = a\times c$ and
$a\preceq d$  iff   $a+d = d$.
By c-semiring properties, $\times$ distributes on $+$. Then, multiplying $c$ on $a+d = d$ we get
$c\times (a+d) = a\times c + c\times d = c\times d$. Hence, $a\times c \preceq c\times d$. 
By using the fact that $b\preceq a\times c$, we conclude $b \preceq c\times d$.
\end{proof}

\paragraph{Proof of Theorem~\ref{thm:sells-cut}}
We first show that Cut permutes over the promotion rule as shown below:
\[
 \infer[Cut]{\nbang{a_1} F_1, \ldots, \nbang{a_n} F_n, \nbang{d_1} G_1, \ldots, \nbang{d_m} G_m \lra \nbang{b} F}
 {
 \infer[\nbang{a}_{R_S}]{\nbang{a_1} F_1, \ldots, \nbang{a_n} F_n \lra \nbang{a} G}
 {\nbang{a_1} F_1, \ldots, \nbang{a_n} F_n \lra G}
 &
 \infer[\nbang{b}_{R_S}]{\nbang{d_1} G_1, \ldots, \nbang{d_m} G_m, \nbang{a} G \lra \nbang{b} F}
 {\nbang{d_1} G_1, \ldots, \nbang{d_m} G_m, \nbang{a} G \lra F}
 }
\qquad \rightsquigarrow\qquad
\]
\[
 \infer[\nbang{b}_{R_S}]{\nbang{a_1} F_1, \ldots, \nbang{a_n} F_n, \nbang{d_1} G_1, \ldots, \nbang{d_m} G_m \lra \nbang{b} F}
 {
 \infer[Cut]{\nbang{a_1} F_1, \ldots, \nbang{a_n} F_n, \nbang{d_1} G_1, \ldots, \nbang{d_m} G_m \lra F}
 {\infer[\nbang{a}_{R_S}]{\nbang{a_1} F_1, \ldots, \nbang{a_n} F_n \lra \nbang{a} G}{\nbang{a_1} F_1, \ldots, \nbang{a_n} F_n \lra  G} 
 & \nbang{d_1} G_1, \ldots, \nbang{d_m} G_m, \nbang{a} G \lra F }
 }
\]
The derivation above is possible since, from the left premise of the first derivation, $a \preceq a_1 \times \cdots \times a_n$
and, from the right premise of the same derivation, $b \preceq a \times d_1 \times \cdots \times d_m$. Thus from the Lemma~\ref{lemma:subst}, 
we have that  $b \preceq a_1 \times \cdots \times a_n \times d_1 \times \cdots \times d_m$, \ie,  the last $\nbang{b}$ can be introduced.

For the rest of the cases, the proof is similar to \sell. The more interesting cases are:
\begin{itemize}
\item Promotion + dereliction
\[
 \infer[Cut]{\Gamma, \Delta \lra F}
 {
 \infer[\nbang{a}_{R_S}]{\Gamma \lra \nbang{a} G}{\Gamma \lra G}
 &
 \infer[\nbang{a}_{L}]{\Delta, \nbang{a} G \lra F}{\Delta, G \lra F}
 }
 \qquad \rightsquigarrow \qquad 
  \infer[Cut]{\Gamma, \Delta \lra F}
 {
{\Gamma \lra G}
 &
 {\Delta, G \lra F}
 }
\]

\item Promotion + weakening
\[
 \infer[Cut]{\Gamma, \Delta \lra F}
 {
 \infer[\nbang{a}_{R_S}]{\Gamma \lra \nbang{a} G}{\Gamma \lra G}
 &
 \infer[\nbang{a}_{L}]{\Delta, \nbang{a} G \lra F}{\Delta \lra F}
 }
 \qquad \rightsquigarrow \qquad 
  \infer=[W]{\Gamma, \Delta \lra F}
 {
 {\Delta \lra F}
 }
\]
We can weaken $\Gamma$ since applying the $\nbang{a}_{R_S}$ rule in the left premise forces
$\Gamma$ to have the shape $\nbang{a_1} F_1, \ldots, \nbang{a_n} F_n$, with $a\preceq a_1\times\ldots\times a_n$. 
On the other hand, from the right-premise, $a \in U$, \ie, 
formulas of the form $\nbang{a} F$ are allowed to contract and weaken. Since $U$ is upwardly closed with respect to $\preceq$, we also have $a_1, \ldots, a_n \in U$. 
Thus  $\nbang{a_1} F_1, \ldots, \nbang{a_n} F_n$
can also be weakened.

\item Promotion + contraction

\resizebox{\textwidth}{!}{
$
 \infer[Cut]{\Gamma, \Delta \lra F}
 {
 \infer[\nbang{a}_{R_S}]{\Gamma \lra \nbang{a} G}{\Gamma \lra G}
 &
 \infer[\nbang{a}_{L}]{\Delta, \nbang{a} G \lra F}{\Delta, \nbang{a} G, \nbang{a} G \lra F}
 }
 \qquad \rightsquigarrow \qquad 
  \infer=[C]{\Gamma, \Delta \lra F}
 {
 \infer[Cut]{\Gamma, \Gamma, \Delta \lra F}
 {
 \infer[\nbang{a}_{R_S}]{\Gamma \lra \nbang{a} G}{\Gamma \lra G}
 &
 \infer[Cut]{\Delta, \Gamma, \nbang{a} G \lra F}{
 \infer[\nbang{a}_{R_S}]{\Gamma \lra \nbang{a} G}{\Gamma \lra G}
 &
 {\Delta, \nbang{a} G, \nbang{a} G \lra F}
 }
 } 
 }
$}

\item  When Cut permutes over structural rules. 

\[
 \infer[Cut]{\nbang{a} H, \Gamma, \Delta \lra F}
 {
 \infer[C]{\nbang{a} H, \Gamma \lra G}{\nbang{a} H, \nbang{a} H, \Gamma \lra G}
 &
 {\Delta, G \lra F}
 }
 \qquad \rightsquigarrow \qquad 
  \infer[C]{\nbang{a} H, \Gamma, \Delta \lra F}
 {
 \infer[Cut]{\nbang{a} H, \nbang{a} H, \Gamma, \Delta \lra F}
 {
 {\nbang{a} H, \nbang{a} H, \Gamma \lra G}
 &
 {\Delta, G \lra F}
 }
 } 
\]

\[
 \infer[Cut]{\nbang{a} H, \Gamma, \Delta \lra F}
 {
 \infer[W]{\nbang{a} H, \Gamma \lra G}{ \Gamma \lra G}
 &
 {\Delta, G \lra F}
 }
 \qquad \rightsquigarrow \qquad 
  \infer[W]{\nbang{a} H, \Gamma, \Delta \lra F}
 {
 \infer[Cut]{\Gamma, \Delta \lra F}
 {
 { \Gamma \lra G}
 &
 {\Gamma, G \lra F}
 }
 } 
\]

\item  Some other principal cases

 \resizebox{\textwidth}{!}{
$
 \infer[Cut]{\Gamma_1, \Gamma_2, \Delta \lra F}
 {
 \infer[\tensor_R]{\Gamma_1, \Gamma_2 \lra A \tensor B}{ \Gamma_1 \lra A & \Gamma_2 \lra B}
 &
 \infer[\tensor_L]{\Delta, A \tensor B \lra F}{\Delta, A, B \lra F}
 }
 \qquad \rightsquigarrow \qquad 
  \infer[Cut]{\Gamma_1, \Gamma_2, \Delta \lra F}
 {
 {\Gamma_1 \lra A}
 &
 \infer[Cut]{ \Gamma_2, \Delta, A \lra F}
 {
 {\Gamma_2 \lra B}
 &
 { \Delta, A, B \lra F}
 }
 }
$
}

\[
 \infer[Cut]{\Gamma, \Delta \lra F}
 {
 \infer[\tensor_R]{\Gamma \lra A \with B}{ \Gamma \lra A & \Gamma \lra B}
 &
 \infer[\with_L]{\Delta, A \with B \lra F}{\Delta, A \lra F}
 }
 \qquad \rightsquigarrow \qquad 
  \infer[Cut]{\Gamma, \Delta \lra F}
 {
 {\Gamma \lra A}
 &
 { \Delta, A \lra F}
 }
\]

\[
 \infer[Cut]{\Gamma, \Delta \lra F}
 {
 \infer[\exists_R]{\Gamma \lra \exists x. G}{\deduce{\Gamma \lra G[t/x]}{\Xi_1}}
 &
 \infer[\exists_L]{\Delta, \exists x. G \lra F}{\deduce{\Delta, G[e/x] \lra F}{\Xi_2}}
 }
 \qquad \rightsquigarrow \qquad 
  \infer[Cut]{\Gamma, \Delta \lra F}
 {
 {\deduce{\Gamma \lra G[t/x]}{\Xi_1}}
 &
 {\deduce{\Delta, G[t/x] \lra F}{\Xi_2[t/e]}}
 }
\]
The proof of the right premise of the right figure, $\Xi_2[t/e]$ is a \sells\ proof using the 
usual eigenvariable argument. This can be proved by induction on the height of proofs.

\end{itemize}

\section{Constraint systems as cylindric algebras}
\label{sec:cylindric}
We shall now recall the  abstract and general definition of constraint systems as cylindric algebras as in \cite{BoerPP95}.

\begin{definition}[Constraint System] \label{def:cs}
A cylindric constraint system  is a structure 
 $
{\bf C} = \langle \cC,\leq,\sqcup,\one,\zero,{\it Var}, \Exists,
D \rangle
$ such that:
\\\noindent{-}
 $\langle \cC,\leq,\sqcup,\one,\zero \rangle$ is a lattice
with $\sqcup$ the $\lub$ operation (representing the logical
\emph{and}), and $\one$, $\zero$ the least and the greatest
elements in $\cC$ respectively (representing $\texttt{true}$ and
$\texttt{false}$). Elements in $\cC$ are called \emph{constraints}
with typical elements $c,c',d,d'...$. If $c\leq d$ and $d\leq c$ we write $c \equivC d$. If $c\leq d$ and $c\not\equivC d$, we write $c<d$. 
\\\noindent{-}${\it Var}$ is a denumerable set of variables.
\\\noindent{-}For each
$x\in {\it Var}$ the function $\Exists x: \cC \to \cC$ is a
cylindrification operator satisfying:
 (E1) $\Exists x (c) \leq c$;  
 (E2) If $c\leq d$ then $\Exists x (c) \leq \Exists x (d)$;
 (E3) $\Exists x(c \sqcup \Exists x (d)) \equivC \Exists x(c) \sqcup \Exists x(d)$;
 (E4) $\Exists x\Exists y(c) \equivC \Exists y\Exists x (c)$.
\\\noindent{-} For each $x,y \in {\it Var}$, the constraint $d_{xy} \in D$ is a
\emph{diagonal element} and it satisfies:
 (D1) $d_{xx} \equivC \one$;
 (D2) If $z$ is different from $x,y$ then $d_{xy} \equivC \Exists z(d_{xz}
\sqcup d_{zy})$;
 (D3) If $x$ is different from $y$ then $c \leq d_{xy} \sqcup
\Exists x(c\sqcup d_{xy})$.
\\\noindent{-} We say that $d$ entails $c$, notation $d \entails c$,  iff $c\leq d$. 
\end{definition}

The cylindrification operators model a sort of existential
quantification, helpful for hiding information. 
 Properties (E1) to (E4) are standard.
 
 The diagonal element $d_{xy}$ can be thought of as the equality $x=y$. Properties (D1) to (D3) are standard and  they allow the definition of  substitutions of the form $[y/x]$  required, for instance, to represent  the substitution of formal and actual parameters in procedure calls. 
 By using these properties, it is easy to prove that $c[y/x]
 \equivC\Exists x.(c \sqcup d_{xy})$,  where $c[y/x]$
  represents abstractly the constraint obtained from $c$ by
replacing the variables $x$ by $y$. As it is customary, we shall assume that the constraint system under consideration contains an equality
theory. Hence,  we shall use indistinguishably the notation $d_{xy}$ and $x = y$ to denote diagonal elements. 
 
\begin{theorem}[Constraint System]
Let  $\mathbb{C} = \langle\cA,  \cC, \entails \rangle$ be as in Definition~\ref{def:soft-cons}. Then, the structure $\langle \cC,\leq,\otimes,\one,\zero, Var, \exists, D \rangle $ is   a cylindric constraint system where $D = \{\nbang{\top_{\cA}}(x=y) \mid  x,y\in Var\}$  and $c\leq d$ iff $d\entails c$. 
\end{theorem}
\begin{proof}
Recall that $c\leq d$ iff the sequent  $\nbang{\top_A}{\delta_1},..., \nbang{\top_A}{\delta_n}, d \lra c$ is provable in SELL
where $\delta_i$ is an axioms in $\Delta$ (see Definition \ref{def:soft-cons}). Abusing of the notation, we shall write sequents as the one above as $\nbang{\top_A}\Delta,d \lra c$. 

Properties (E1) to (E4) of $\Exists$ (interpreted as $\exists$) are easy. 

Note that the constraint system contains an equality theory and then,   $\Delta$ define the meaning of ``$=$''. Observe also that diagonal elements are marked with the largest subexponential $\top_{\cA}$ (which is unbounded). Then, it is easy to see that the following sequents are provable:
$\nbang{\top_{\cA}}\Delta \lra \nbang{\top_{\cA}}(x=x) \equiv \one $;
$\nbang{\top_A}\Delta \lra \nbang{\top_{\cA}}(x=y) \equiv \exists z.(\nbang{\top_{\cA}}(x=z) \otimes \nbang{\top_{\cA}}(z=y))$ whenever $z$ is different from $x$ and $y$; and $\nbang{\top_{\cA}}\Delta, \nbang{\top_{\cA}}(x=y), \exists x.(c \otimes \nbang{\top_{\cA}}(x=y)) \lra c$ if $x$ is different from $y$. Then, properties $(D1)$ to $(D3)$ hold. 

Finally, we note that according to Definition \ref{def:soft-cons},  every constraint $c$ is a \emph{classical} formula. Then it follows that  for any $c,d$, the sequents $\nbang{\top_{\cA}}{\Delta}, c \lra \one$, $\nbang{\top_{\cA}}{\Delta}, \zero \lra$ $c$ and
$\nbang{\top_{\cA}}\Delta, c,d \lra c$ are also provable. This shows that indeed $\langle \cC,\leq,\otimes,\one,\zero \rangle$ is a lattice where $\otimes$ is the lub and $\one$ (resp. $\zero$) the least (resp. greatest) element. 
\end{proof}

\end{document}